\documentclass[12pt]{amsart}

\usepackage{amsmath,amssymb,amsthm,latexsym,cite}
\usepackage[dvips]{graphicx}

\newtheorem{theorem}{Theorem}[section]
\newtheorem{lemma}[theorem]{Lemma}
\newtheorem{corollary}[theorem]{Corollary}

\newtheorem{proposition}[theorem]{Proposition}
\theoremstyle{definition}
\newtheorem{definition}[theorem]{Definition}
\theoremstyle{remark}
\newtheorem{remark}[theorem]{Remark}

\newcommand{\Graph}{\Gamma}
\newcommand{\G}{\Gamma}

\newcommand{\Hamiltonian}{\mathcal{H}}
\newcommand{\Domain}{\mathcal{D}}
\newcommand{\Id}{\mathbb{I}}
\newcommand{\cA}{\mathcal{A}}
\newcommand{\C}{\mathbb{C}}
\newcommand{\cT}{\mathcal{T}}
\newcommand{\Vertices}{\mathcal{V}}
\newcommand{\Edges}{\mathcal{E}}
\newcommand{\Lengths}{\mathcal{L}}
\newcommand{\CC}{\mathbb{C}}
\newcommand{\R}{\mathbb{R}}
\newcommand{\cH}{\mathcal{H}}
\newcommand{\rmi}{{\mathrm{i}}}

\newcommand\term[1]{{\em #1\/}}
\newcommand\bbar[1]{\overline{#1}}

\newcommand{\HkG}[1]{{\widetilde{H}^{#1}(\G)}}

\DeclareMathOperator{\diag}{diag}
\DeclareMathOperator{\Real}{\mathfrak{Re}}

\title[Spectra of quantum
graphs]{Dependence of the spectrum of a quantum graph on vertex
  conditions and edge lengths
}

\author{Gregory~Berkolaiko}
\author{Peter~Kuchment}

\address{Dept. of Mathematics, Texas A\&M University, College Station,
  TX 77843-3368, USA}

\begin{document}


\begin{abstract}
  We study the dependence of the quantum graph Hamiltonian, its
  resolvent, and its spectrum on the vertex conditions and edge
  lengths. In particular, several results on the analyticity and
  interlacing of the spectra of graphs with different vertex
  conditions are obtained and their applications are discussed.
\end{abstract}

\maketitle

\section{Introduction}
\label{sec:intro}

Graph models have long been used as a simpler setting to study
complicated phenomena.  Quantum graphs in particular have recently
gained popularity as models for thin wires, eigenvalue statistics of
chaotic systems and properties of the nodal domains of
eigenfunctions.  We refer the interested reader to the recent reviews
\cite{Kuc_wrm02,GnuSmi_ap06,Kuc_incol08} and collections of papers
\cite{BerCarFulKuc_eds06,ExnKeaKuc_eds08}.

A quantum graph is a metric graph equipped with a self-adjoint
differential ``Hamiltonian'' operator (usually of Schr\"odinger type)
defined on the edges and matching conditions specified at the
vertices.  Every edge of the graph has a length assigned to it.  In
this manuscript we establish several results concerning the general
properties of the spectrum of the Hamiltonian as a function of the
parameters involved: the edge lengths and matching (vertex)
conditions.  In Section~\ref{sec:setup} we introduce the standard
notions related to quantum graphs, vertex conditions, and quadratic
forms of the corresponding Hamiltonians.  The most widely used example
of the vertex conditions, the $\delta$-type condition, is described in
some detail.  The results presented in Section~\ref{sec:dependence}
concern the analyticity of the spectrum as a function of the
parameters involved: vertex conditions and edge lengths.  Analytic
dependence on the potential (as an infinite-dimensional parameter) can
also be established by the methods used in the Section, but we omit
this to keep the presentation simple.  In Section~\ref{sec:Hadamard}
we compute the derivative of an eigenvalue with respect to a variation
in the length of an edge; this is an analogue of a well-known Hadamard
variational formula.  Derivative with respect to a parameter of a
$\delta$-type condition is also computed.

Section~\ref{sec:eig_interlacing} again focuses on the $\delta$-type
condition at a vertex.  Varying such a condition at a specified vertex
or changing connectivity of the vertex we obtain families of spectra
and study their interlacing properties.  Eigenvalue interlacing (or
bracketing) is a powerful tool in spectral theory with such well-known
applications as the derivation of the asymptotic Weyl law, see
\cite{CourantHilbert_volume1}.  In the graph setting, it allows one to
estimate eigenvalue of a given graph via the eigenvalues of its
subgraphs, which may be easier to calculate.  Interlacing results on
graphs have already been used in several situations
\cite{Sch_wrcm06,Ber_cmp08,LlePos_jmaa08}.  We significantly
generalize these results and put them in the form particularly suited
for the applications.  We discuss several applications.  In
particular, we give a simple derivation of the number of nodal domains of
the $n$-th eigenfunction on a quantum tree and study irreducibility of
the spectrum of a family of graphs obtained by varying $\delta$-type
condition at one of the vertices.

\section{Quantum graph Hamiltonian}
\label{sec:setup}

Let $\Graph = (\Vertices, \Edges)$ be a graph with finite sets of
vertices $\Vertices=\{v_j\}$ and edges $\Edges=\{e_j\}$. We will
assume that $\G$ is a \term{metric graph}, i.e. every edge $e$ is a
$1$-dimensional segment with a positive finite length $L_e$ and a
coordinate $x_e$ assigned.  Each edge corresponds to two
\term{directed edges}, or \term{bonds}, of opposite directions.  If
bonds $b_1$ and $b_2$ correspond to the same edge, they are called
\term{reversals} of each other. In this case we use the notations $b_1
= \bbar{b_2}$ and $b_2 = \bbar{b_1}$.  A bond $b$ inherits its length
from the edge $e$ it corresponds to, in particular, $L_{b} =
L_{\bbar{b}}=L_e$.  A coordinate $x_b$ is assigned on each bond, and
the coordinates on mutually reversed bonds are connected via
$x_{\bbar{b}} = L_b-x_b$.

\begin{definition}
  \label{D:spaces}
  \indent
  \begin{itemize}
  \item The \term{space $L_2(\Gamma )$}
    on $\Gamma $ consists of functions that are measurable and square
    integrable on each edge $e$ with the norm
    \begin{equation*}\label{E:L2_sum}
      \|f\|_{L_2(\Gamma)}^2:=\sum\limits_{e \in \Edges}\|f\|_{L_2(e)}^2.
    \end{equation*}
    In other words, $L_2(\Gamma )$ is the orthogonal direct sum of
    spaces $L_2(e)$.

  \item We denote by $\widetilde{H}^k(\G)$ the space
    \begin{equation*}\label{E:tilde_space}
      \widetilde{H}^k(\G):=\bigoplus_{e\in\Edges}H^k(e),
    \end{equation*}
    which consists of the functions $f$ on $\G$ that on each edge $e$
    belong to the Sobolev space $H^k(e)$ and is equipped with the norm
    \begin{equation*}\label{E:tilde_space_sum}
      \|f\|_{\HkG{k}}^2 := \sum\limits_{e\in\Edges}\|f\|^2_{H^k(e)}.
    \end{equation*}

  \item The \term{Sobolev space $H^1(\Gamma)$}
    consists of all {\bf continuous} functions from $\HkG{1}$.
 \end{itemize}
\end{definition}

Note that in the definition of $\widetilde{H}^k(\G)$ the smoothness is
enforced along edges only, without any junction conditions at the
vertices at all.  The continuity condition imposed on functions from
the Sobolev space $H^1(\G)$ means that any function $f$ from this
space assumes the same value at a vertex $v$ on all edges adjacent to
$v$, and thus $f(v)$ is uniquely defined. This is a natural condition
for one-dimensional $H^1$-functions, which are known to be continuous
in the standard $1$-dimensional setting.

A metric graph becomes quantum after being equipped with an additional
structure: assignment of a self-adjoint differential operator.  This operator will
be also called the \term{Hamiltonian}.  The frequently arising
in the quantum graph studies operator is the negative second
derivative acting on each edge ($x$ is the coordinate $x$ along an edge)
\begin{equation}
  f(x) \mapsto -\frac{d^2f}{dx^2}.  \label{E:deriv_2}
\end{equation}
or the more general Schr\"{o}dinger operator
\begin{equation}
  f(x) \mapsto -\frac{d^2f}{dx^2}+V(x)f(x),  \label{E:electr}
\end{equation}
where $V(x)$ is an \term{electric potential}. (Our results will hold, for instance, for $V\in L_2(\G)$. Generalizations to more general operators, e.g. including magnetic terms, are also straightforward.)

Notice that for both these operators the direction of the edge is
irrelevant. This is not true anymore if one wants to include
derivative term of an odd order, e.g. magnetic potential, but we shall
not address such operators in the present note (see, e.g.,
\cite{FulKucWil_jpa07,Pos_AHP09} concerning these issues).

The natural smoothness requirement coming from the ODE theory is that
$f$ belongs to the Sobolev space $H^2(e)$ on each edge
$e$.  Appropriate boundary value conditions at the vertices
(\term{vertex conditions}) still need to be added, which are considered in the next subsection.

\subsection{Vertex conditions.}
\label{sec:vertex_conditions}

We will briefly describe now the known descriptions of the vertex
conditions that one can add to the differential expression
(\ref{E:electr}) in order to create a self-adjoint operator (see,
e.g., \cite{KosSch_jpa99,Har_jpa00,Kuc_wrm04,FulKucWil_jpa07} for
details).

Assume that the domain of the operator is a subspace of the Sobolev
space $\HkG{2}$ (see the references above for the justification of
this assumption). Then the standard Sobolev trace theorem (e.g.,
\cite{EdmundsEvans_spectral}) implies that both the function $f$,
$f\in H^2(e)$, and its first derivative have correctly defined values
at the endpoints of the edge $e$.  Thus, for a function $f\in \HkG{2}$
and a vertex $v$ we can define the column vectors $F(v)$ and $F'(v)$
\begin{equation}\label{E:F-column}
  F(v) :=
  \left(
    \begin{array}{c}
      f_{e_1}(v) \\
      \dots \\
      \dots \\
      f_{e_{d_v}}(v) \\
    \end{array}
  \right)
  \qquad
  F'(v) :=
  \left(
    \begin{array}{c}
      f_{e_1}'(v) \\
      \dots \\
      \dots \\
      f_{e_{d_v}}'(v) \\
    \end{array}
  \right)
\end{equation}
of the values at the vertex $v$ that functions $f$ and $f'$ attain
along the edges incident to $v$.  Here $d_v$ is the degree of the
vertex $v$.  The derivatives of $f$ at vertices are always taken away
from the vertices and into the edges.

Descriptions of all possible vertex conditions that would make
operator $\Hamiltonian$ self-adjoint can be done in several somewhat
different ways.  Below we list the most usual descriptions, which were
introduced in \cite{KosSch_jpa99,Har_jpa00,Kuc_wrm04}.

\begin{theorem}\label{T:BC} Let $\Gamma$ be a metric graph with
  finitely many edges. Consider the operator ${\mathcal H}$ acting as
  $-\dfrac{d^2}{dx_e^2}+V(x)$ on each edge $e$, with the domain consisting
  of functions that belong to $H^2(e)$ and satisfying some local vertex
  conditions involving vertex values of functions and their
  derivatives. The operator is self-adjoint if and only if the vertex
  conditions can be written in one (and thus any) of the following
  three forms:
  \begin{description}
  \item[A] For every vertex $v$ of degree $d_v$ there exist $d_v \times
    d_v$ matrices $A_v$ and $B_v$ such that
    \begin{equation}
      \label{E:condit_rank}
      \mbox{The } d_v\times 2d_v \mbox{ matrix }
      (A_v\,B_v) \mbox{ has the maximal rank.}
    \end{equation}
    \begin{equation}
      \label{E:condit_sa}
      \mbox{The matrix } A_vB_v^* \mbox{ is self-adjoint.}
    \end{equation}
    and functions $f$ from the domain of $\Hamiltonian$ satisfy the vertex conditions
    \begin{equation}
      \label{E:cond_kossch}
      A_vF(v)+B_vF'(v)=0.
    \end{equation}

  \item[B] For every vertex $v$ of degree $d_v$, there exists a
    unitary $d_v\times d_v$ matrix $U_v$ such that functions $f$ from
    the domain of $\Hamiltonian$ satisfy the vertex conditions
    \begin{equation}\label{E:cond_harmer}
      \rmi(U_v-\Id)F(v)+(U_v+\Id)F^\prime(v)=0,
    \end{equation}
    where $\mathbb{I}$ is the $d_v\times d_v$ identity matrix.

  \item[C] For every vertex $v$ of degree $d_v$, there are three
    orthogonal (and mutually orthogonal) projectors $P_{D,v}$,
    $P_{N,v}$ and $P_{R,v} := \Id-P_{D,v}-P_{N,v}$ (one or two
    projectors can be zero) acting in $\mathbb{C}^{d_v}$ and an
    invertible self-adjoint operator $\Lambda_v$ acting in the
    subspace $P_{R,v}\mathbb{C}^{d_v}$, such that functions $f$ from
    the domain of $\Hamiltonian$ satisfy the vertex conditions
    \begin{equation}
      \begin{cases}
        P_{D,v}F(v)=0 \mbox{ - the ``Dirichlet part''}, \\
        P_{N,v}F'(v)=0 \mbox{ - the ``Neumann part''},\\
        P_{R,v}F'(v) =\Lambda_v P_{R,v}F(v) \mbox{ - the ``Robin part''}.\\
      \end{cases}\label{E:cond_kuch}
    \end{equation}
  \end{description}
\end{theorem}

\subsection{Quadratic form}
\label{sec:quadratic_form}

To describe the quadratic form of the operator ${\mathcal H}$, which
is a self-adjoint realization of the Schr\"odinger operator
(\ref{E:electr}) acting along each edge, the self-adjoint vertex
conditions written in the form (C) of Theorem \ref{T:BC} are the most
convenient.  The following theorem is cited from \cite{Kuc_wrm04}.

\begin{theorem}
  \label{T:qform}
  The quadratic form $h$ of ${\mathcal H}$ is given as
  \begin{equation}\label{E:qform}
    h[f,f]= \sum\limits_{e \in \Edges} \int\limits_e
    \left|\frac{df}{dx}\right|^2dx
    + \sum\limits_{e \in \Edges} \int\limits_e
    V(x) \left|f(x)\right|^2dx
    + \sum\limits_{v \in \Vertices} \left\langle
      \Lambda_vP_{R,v}F,P_{R,v}F \right\rangle,
  \end{equation}
  where $\langle , \rangle$ denotes the standard Hermitian inner
  product in $\CC^{\dim P_{R,v}}$. The domain of this form consists of all
  functions $f$ that belong to $H^1(e)$ on each edge $e$ and satisfy
  at each vertex $v$ the condition $P_{D,v}F=0$.

  Correspondingly, the sesqui-linear form of ${\mathcal H}$ is
  \begin{equation}\label{E:sesqform}
    h[f,g]= \sum\limits_{e \in \Edges} \int\limits_e
    \frac{df}{dx}\overline{\frac{dg}{dx}}dx
    + \sum\limits_{e \in \Edges} \int\limits_e
    V(x) f(x) \overline{g(x)} dx
    + \sum\limits_{v\in\Vertices} \left\langle
      \Lambda_vP_{R,v}F,P_{R,v}G \right\rangle.
  \end{equation}
\end{theorem}

\subsection{Examples of vertex conditions}
\label{sec:conditions}

In this paper we will often be dealing with the \term{ $\delta$-type
  conditions} which are defined at a vertex $v$ as follows:
\begin{equation}
  \label{E:delta_cond}
  \left\{
    \begin{array}{l}
      f(x)\text{ is continuous at }v, \\[.5em]
      \sum_{e \in \Edges_v}\frac{df}{dx_e}(v)=\alpha _v f(v),
    \end{array}
  \right.
\end{equation}
where for each vertex $v$, $\alpha_v $ is a fixed number. One recognizes this condition as being an analog of the conditions one obtains for the Schr\"{o}dinger
operator on the line with a $\delta$ potential. The special case $\alpha_v=0$ is known as the Neumann (or
Kirchhoff) condition.

The $\delta$-type condition can be written in the form
(\ref{E:cond_kossch}) with
\begin{displaymath}
  A_v=\begin{pmatrix}
    1 & -1 & 0 & .... & 0 \\
    0 & 1 & -1 & ... & 0 \\
    \dots & \dots & \dots & \dots & \dots \\
    ... & ... & 0 & 1 & -1 \\
    \alpha_v & 0 & ... & 0 & 0
  \end{pmatrix}
\end{displaymath}
and
\begin{displaymath}
  B_v=\begin{pmatrix}
    0 & 0 &  .... & 0 \\
    ... & ...& ... & ... \\
    0 & 0 &  .... & 0 \\
    1 & 1 & ... & 1
  \end{pmatrix}.
\end{displaymath}
Since
\begin{displaymath}
  A_vB_v^*=\begin{pmatrix}
    0 &   \dots & 0 & 0 \\
    \dots  & \dots & \dots & \dots \\
    0  & \dots & 0 & 0 \\
    0 &  \dots & 0 & \alpha
  \end{pmatrix},
\end{displaymath}
the self-adjointness condition (\ref{E:condit_sa}) is satisfied if and
only if $\alpha_v$ is real.

In order to write the vertex conditions in the form
(\ref{E:cond_kuch}), one introduces the orthogonal projection
$P_{D,v}$ onto the kernel of $B_v$, the projector
$P_{R,v}=\Id-P_{D,v}$ and the self-adjoint operator
$\Lambda_v=B_v^{(-1)}A_v$ on the range of $P_{R,v}$.  A
straightforward calculation shows that $P_{R,v}$ is the
one-dimensional orthogonal projector onto the space of vectors with
equal coordinates and thus the range of $P_{D,v}$ is spanned by the
vectors $r_k$, $k=1,...,d_v-1$, where $r_k$ has $1$ as the $k$-th
component, $-1$ as the next one, and zeros otherwise. Then $\Lambda_v$
becomes the multiplication by the number $\dfrac{\alpha_v}{d_v}$.  In
particular, the quadratic form of the operator $\Hamiltonian$ (assuming
$\delta$-type conditions on all vertices of the graph) is
\begin{align}
  \label{E:qform_delta}
  \nonumber
  h[f,f] &=
  \sum_{e \in \Edges} \int_e \left|\frac{df}{dx}\right|^2dx     
  + \sum\limits_{e \in \Edges} \int\limits_e
  V(x) \left|f(x)\right|^2dx
  + \sum_{v \in \Vertices} \langle L_v F, F \rangle\\
  &=\sum_{e \in \Edges} \int_e \left|\frac{df}{dx}\right|^2dx 
  + \sum\limits_{e \in \Edges} \int\limits_e
  V(x) \left|f(x)\right|^2dx
  + \sum_{v \in \Vertices} \alpha_v|f(v)|^2,
\end{align}
defined on $f \in H^1(\Gamma)$, which are automatically continuous,
and so $F(v)=(f(v),...,f(v))^t$.

\term{Vertex Dirichlet condition}
requires that the function vanishes at the vertex: $f(v)=0$. At the first glance, it might look like it is significantly different from the $\delta$-type conditions, but a closer inspection shows that this is not the case. Indeed, since the function must vanish when approaching the vertex from any edge, the vertex Dirichlet condition can be recast in the following form:
\begin{equation}
  \label{E:dirichlet_cond}
  \left\{
    \begin{array}{l}
      f(x)\text{ is continuous at }v, \\[.5em]
      f(v) = 0,
    \end{array}
  \right.
\end{equation}
Now one finds resemblance with (\ref{E:delta_cond}), and indeed, if one divides the equality in  (\ref{E:delta_cond}) by $\alpha_v$ and then takes the limit when $\alpha_v\to\infty$, one arrives to (\ref{E:dirichlet_cond}).

Hence, vertex Dirichlet condition seems to be the limit case of  (\ref{E:delta_cond}) when $\alpha_v\to\infty$. We thus introduce the \term{extended $\delta$-type conditions\/} by allowing $\alpha_v=\infty$. In order to avoid considering infinite values of $\alpha_v$, the two types of conditions can be also written in the form
\begin{equation}\label{E:dirichlet_sin}
  \cos(\gamma_v)\sum_{e \in \Edges_v} \frac{df}{dx_e}(v)
  = \sin(\gamma_v) f(v).
\end{equation}
Here $\gamma_v=0$ corresponds to the Neumann condition and $\gamma_v =
\pi/2$ corresponds to the Dirichlet one, with more general
$\delta$-type conditions in between. Usefulness of considering
Dirichlet condition as a part of the family of $\delta$-type
conditions becomes clear in spectral theory, as will be illustrated,
for instance, in Theorem~\ref{thm:interlacing}.

In fact, it will be convenient later on in this paper to rewrite the extended $\delta$-type conditions (\ref{E:dirichlet_sin}) in the following form:
\begin{equation}\label{E:extended_compl}
  (z+1)\sum_{e \in \Edges_v} \frac{df}{dx_e}(v)
  = \rmi (z-1) f(v),
\end{equation}
where $z$ belongs to the unit circle in the complex plane, i.e. $|z|=1$.

Interpreting the Dirichlet condition in terms of the corresponding projectors, as in part C of Theorem \ref{T:BC}, one notices that here $P_{D,v}=\Id$ and, correspondingly, $P_{R,v}=0$.  Hence there is no
additive contribution to the quadratic form $h[f,f]$ coming from the vertex $v$.  Instead, the condition $f(v)=0$ is introduced directly into the domain $D(h)$.

To summarize, the quadratic form for a graph with the  extended $\delta$-type conditions (i.e., allowing $\alpha_v=\infty$) at all vertices, can be
written as
\begin{equation}\label{E:quadr}
  h[f,f] = \sum_{e\in \Edges} \int_e \left|\frac{df}{dx}\right|^2dx
  + \sum\limits_{e \in \Edges} \int\limits_e
  V(x) \left|f(x)\right|^2dx
  + \sum_{\{v \in \Vertices\,|\,\alpha_v < \infty\}} \alpha_v
  \left|f (v) \right|^2,
\end{equation}
where
\begin{displaymath}
  f \in \widetilde{H}^1(\G),\ f \mbox{ is
    continuous on } \Graph, \mbox{ and }
  f(v) = 0 \mbox{ whenever } \alpha_v=\infty.
\end{displaymath}

Vertex Dirichlet condition is an example of a \term{decoupling
  condition}, since it essentially removes any connection between the
edges attached to the vertex.  Another example that will be useful to
us is $P_{R,v} = \Id$, $P_{N,v} = P_{D,v} = 0$ and $\Lambda_v =
\diag(\alpha_1,\ldots, \alpha_{d_v})$.  In this case the function is
no longer required to be continuous at the vertex and the condition
reduces to
\begin{equation*}
  f'_e(v) = \alpha_e f_e(v)
\end{equation*}
on every edge $e$ incident to the vertex $v$.

\section{Dependence on vertex conditions and edge lengths}
\label{sec:dependence}

\subsection{Dependence of the Hamiltonian}
\label{sec:dependence_H}

In this section, we discuss the (analytic) dependence of the quantum
graph Hamiltonian $\Hamiltonian$ on the vertex conditions and the edge
lengths. This issue happens to be important in many circumstances,
e.g. when considering dependence of the spectrum and the
eigenfunctions.  To keep the notation simpler, we only address the
case $V\equiv 0$.  However, the results can be extended to non-zero
electric potentials without any change in the proofs. In fact, the
dependence on the potential is also analytic, so the potential can be
added to the vertex conditions and edge length as an extra
(infinite-dimensional) parameter.

As before, the graph $\G$ is assumed to be finite (i.e., it has finitely many
vertices and finite lengths edges).

It will be convenient here to consider the vertex conditions in the
form of (\ref{E:cond_kossch}):
\begin{equation*}
 A_vF(v)+B_vF'(v)=0.
\end{equation*}

Given a set of matrices $\cA_v=(A_v\, B_v)|_{v\in\Vertices}$, we denote by $\cA$ the collection of $\cA_v$ for all $v\in\Vertices$:
\begin{equation*}
    \cA:=\{\cA_v\}|_{v\in\Vertices}.
\end{equation*}
Then $\cA$ can be considered as a block-diagonal matrix of the size $(\sum_v d_v)\times 2(\sum_v d_v)$, with individual blocks of sizes $d_v \times 2d_v$.
The space of such complex matrices $\cA$ can be identified with $\CC^{2\sum_v d_v^2}$, or just $\CC^{2q}$, where we will use the shorthand notation
\begin{equation*}
    q:=\sum_v d_v^2.
\end{equation*}
We now define the sub-set $U$ of $\CC^{2q}$ that consists of matrices
$\cA$ satisfying the maximal rank condition (\ref{E:condit_rank}) for
each $v\in\Vertices$:
\begin{equation*}
    U:=\{\cA\,\mid\,(A_v\, B_v)\mbox{ has maximal rank for any }v\in\Vertices\}.
\end{equation*}
We denote by $U_s$ the subset of $U$ consisting of matrices satisfying the self-adjointness condition (\ref{E:condit_sa}):

\begin{equation*}
    U_s:=\{\cA\in U \mid\,A_vB_v^*\mbox{ is self-adjoint for any }v\in\Vertices\}.
\end{equation*}

The following statement describes some simple properties of these sets:
\begin{lemma}\label{L:A_regular}
\indent
\begin{enumerate}
\item The complement $\CC^{2q}\setminus U$ of $U$ in $\CC^{2q}$ is algebraic.
\item $U$ is an open and everywhere dense domain of holomorphy in $\CC^{2q}$.
\end{enumerate}
\end{lemma}
\begin{proof}
Algebraicity of the set $\CC^{2q}\setminus U$ is clear, since its elements are described by the algebraic relations forcing the highest order minors to vanish. This proves the first statement of the Lemma. The second claim is an immediate corollary of the first one, if one can show that $U$ is not empty. This is done by noticing that $\cA:=(\Id \quad 0)\in U$.
\end{proof}
Let us now return to the quantum graph Hamiltonian. Since we are going to look into its dependence on vertex conditions, we introduce the corresponding notation:
\begin{definition}
  The operator $-d^2/dx^2$ defined on functions from
  $\widetilde{H}^2(\G)$ satisfying (\ref{E:cond_kossch}) at each
  vertex $v$ is denoted by $\Hamiltonian_{\cA}$.
\end{definition}

As we have already seen, the domain $\Domain(\Hamiltonian_{\cA})$ of the operator $\Hamiltonian_{\cA}$ is a closed subspace $\Domain_{\cA}$ of $\widetilde{H}^2(\G)$ that can be described as follows:
\begin{equation*}\label{E:A_domain}
    \Domain_{\cA}:=\{f\in \widetilde{H}^2(\G)\mid A_vF(v)+B_vF'(v)=0 \mbox{ for all }v\in\Vertices\}.
\end{equation*}
In other words, $\Domain_{\cA}$ is the kernel of the continuous linear operator
\begin{equation*}
\cT_{\cA}:\widetilde{H}^2(\G) \to \CC^{2E},\end{equation*}
where
\begin{equation}\label{E:boundary_op_T}
    \cT_{\cA}: f \mapsto \{A_vF(v)+B_vF'(v)\}\in\bigoplus\limits_{v\in\Vertices}\CC^{d_v}=\CC^{2E}.
\end{equation}
This simple observation allows us to establish nice behavior of the
domain of $\Hamiltonian_{\cA}$ with respect to the vertex conditions
matrix $\cA$. In order to do this, let us start with a simple lemma:
\begin{lemma}\label{L:T_A}\indent
\begin{enumerate}
\item The operator function $\cA\mapsto \cT_{\cA}$ is analytic in $\CC^{2q}$ with values in the space $L(\widetilde{H}^2(\G), \CC^{2E})$ of bounded linear operators from $\widetilde{H}^2(\G)$ to $\CC^{2E}$.
\item For any $\cA\in U$, the operator $\cT_{\cA}$ is surjective.
\end{enumerate}
\end{lemma}
\begin{proof} Indeed, according to (\ref{E:boundary_op_T}),  the function is in fact linear, and thus analytic with respect to $\cA$. Also, the set of vectors $\{F(v),F'(v)\}|_{v\in\Vertices}$ achievable from elements $f$ of $\oplus_e H^2(e)$ is clearly arbitrary. Then, if the maximal rank condition is satisfied, this implies the surjectivity of $\cT_{\cA}$.
\end{proof}
Let us consider the trivial vector bundle
\begin{equation*}
U\times \widetilde{H}^2(\G) \to U
\end{equation*}
over $U$ with fibers equal to $\widetilde{H}^2(\G)$. Consider the sub-set
\begin{equation*}
\Domain:=\bigcup\limits_{\cA\in U} \left(\{\cA\} \times \Domain_{\cA}\right) \subset U\times \widetilde{H}^2(\G).
\end{equation*}
In other words, we look at the domain $\Domain_{\cA}$ of the operator $\Hamiltonian_{\cA}$ as a ``rotating'' with $\cA$ subspace of $\widetilde{H}^2(\G)$. The next results shows that this domain rotates ``nicely'' (analytically) with $\cA$ and is topologically and analytically trivial as a vector bundle.
\begin{theorem}\label{T:domain_A}
\indent
\begin{enumerate}
\item $\Domain$ is an analytic sub-bundle of co-dimension $2E$ of the ambient trivial bundle.
\item The bundle $\Domain$ is trivializable. In other words, there exists a trivialization, i.e. an analytic operator-function $T(\cA)$ on $U$ with values in linear bounded operators from a Hilbert space $H$ into $\widetilde{H}^2(\G)$, such that $\ker T(\cA)=0$ and $\mathrm{ran}\, T(\cA)=\Domain_{\cA}$ for any $\cA\in U$.
\item After the trivialization, the values of the resulting analytic in $U$ operator-function
    \begin{equation*}
    \cA\mapsto \Hamiltonian_{\cA}T(\cA)
    \end{equation*}
    are Fredholm operators of index zero from $H$ to $L^2(\G)$.
\end{enumerate}
\end{theorem}
\begin{proof} The first statement of the Theorem is local. So, let us pick a matrix $\cA_0\in U$. According to Lemma \ref{L:T_A}, the operator $\cT_{\cA_0}$ is surjective, and thus has a continuous right inverse $R$. Thus, $\cT_{\cA_0}R=\Id$. Then $\cT_{\cA}R$ is invertible for $\cA$ close to $\cA_0$. This implies that for such $\cA$, one has $\cT_{\cA}R(\cT_{\cA}R)^{-1}=\Id$. This implies that
\begin{equation*}
P(\cA):=\Id-R(\cT_{\cA}R)^{-1}\cT_{\cA}
\end{equation*}
is a projector onto the kernel of $\cT_{\cA}$, i.e. on $\Domain_{\cA}$. Since this projector, by construction, is analytic with respect to $\cA$ in a neighborhood of $\cA_0$, this proves a part of the first claim of the theorem:  $\Domain$ is an analytic Hilbert sub-bundle (e.g., \cite{ZaiKreKuc_umn75}). It only remains to notice that co-dimension of the kernel of $\cT_{\cA}$ is the dimension of its range, i.e. $2E$.

To prove the second claim, we notice that $\Domain$ is an infinite
dimensional analytic Hilbert bundle. Due to the Kuiper's theorem on
contractibility of the general linear group of any
infinite-dimensional Hilbert space \cite{Kui_t65}, all such bundles
are topologically trivial. Since the base $U$ is holomorphically
convex, the Bungart's theorem \cite{Bun_t67} says that the same
holds in the analytic category (see further discussion of the technique and relevant references in the survey  \cite{ZaiKreKuc_umn75}).

Let us prove the third claim. Since $\Hamiltonian_{\cA}$ is the restriction to $\Domain_{\cA}$ of the fixed operator $-d^2/dx^2$ (with no vertex conditions attached), acting continuously from $\widetilde{H}^2(\G)$ to $L^2(\G)$, the
operator-function in question can be written as $(-d^2/dx^2)T(\cA)$ and thus is analytic. The Fredholm property and the zero value of the index follow from \cite[Theorem 14]{FulKucWil_jpa07}.
\end{proof}

\subsection{Dependence of the resolvent}
\label{SS:resolvent}

We now consider the question about the dependence on vertex conditions of the resolvent of $\Hamiltonian_{\cA}$. In order to do so, we need to consider
the operator family $\Hamiltonian_{\cA}-i \Id$, where $\Id$ denotes the identity operator in $L^2(\G)$. Sobolev's compactness of embedding theorem shows that $\Id$ is a compact operator from $\widetilde{H}^2(\G)$ to $L_2(\G)$. This and the previous theorem imply that $\Hamiltonian_{\cA}-i \Id$ is also an analytic family of Fredholm operators of zero index.

As analytic Fredholm theorem (see Appendix \ref{A:Fredholm}) shows, the set $\Sigma$ of the matrices $\cA$ for which the operator $\Hamiltonian_{\cA}-i \Id:\Domain_{\cA}\to L^2(\G)$ is not continuously invertible, is principal analytic (i.e., can be given by an equation $\phi(\cA)=0$, where function $\phi$ is analytic). Since this set $\Sigma$ does not include any points $\cA$ of $U_s$ (because in this case the operator $\Hamiltonian_{\cA}$ is self-adjoint), this singular set is nowhere dense.
\begin{theorem}\label{T:A_resolvent}
The resolvent $(\Hamiltonian_{\cA}-i \Id)^{-1}$, defined in $U\setminus \Sigma$, is analytic and has values that are compact operators in $L^2(\G)$.
\end{theorem}
\begin{proof}
  Analyticity of the inverse to an analytic family of bounded
  invertible operators is well known (e.g.,
  \cite{Kato_perturbation,ZaiKreKuc_umn75}). Thus, $(\Hamiltonian_{\cA}-i
  \Id)^{-1}$ is an analytic family of operators from $L^2(\G)$ to
  $\widetilde{H}^2(\G)$, and thus also as a family of operators acting in
  $L^2(\G)$. However, considered as operators in $L^2$, they
  factor through the compact embedding of $\Domain_{\cA}$ into
  $L^2(\G)$, and thus are compact.
\end{proof}

\subsection{Variations in the edge lengths}
\label{SS:edge_lengths}

Sometimes one needs to consider the quantum graph's dependence on the
variations in the edge lengths parameters $\{L_e\}$ (without changing
graph's topology). We thus extend the previous considerations to
include the dependence on the vector
\begin{equation*}
\Lengths:=\{L_e\}|_{e\in\Edges}\in (\R^+)^E.
\end{equation*}
in fact, we allow ``complex values of lengths'' $\Lengths$, i.e. $\Lengths\in (\CC\setminus\{0\})^E$.

Let $\G$ be a quantum graph with the Hamiltonian $\Hamiltonian=-d^2/dx^2$ equipped with vertex conditions described by a matrix $\cA=\{(A_v\, B_v)|_{v\in\Vertices}\}$. We would like to vary the lengths of the edges (independently from each other), without changing topology or vertex conditions. Let us consider the vector $\xi=\{\xi_e\}\in (\R^+)^E$ of dilation factors along each edge. It is clear that such a dilation is equivalent to keeping the metric graph structure the same, while replacing $\Hamiltonian=-d^2/dx^2$ with $\Hamiltonian_\xi=\{-\xi^{-2}_e d^2/dx^2$ and $(A_v\, B_v)$ with $(A_v\ B_v\Xi_v)$, where $\Xi_v$ is the diagonal $d_v\times d_v$ matrix having $\xi_e$ as its diagonal entries, where $e$ denotes the edges incident to the vertex $v$.

\begin{definition}\label{D:Hamilt_scaled}
For any $\cA$ from the previously described set $U$ of complex matrices $\cA=\{(A_v\, B_v)\}$ satisfying the maximal rank condition and for any vector $\xi\in(\CC\setminus\{0\})^E$, we denote by $\Hamiltonian_{\cA,\xi}$
the Hamiltonian $-\xi_e^{-2}d^2/dx_e^2$ on $\G$ with the vertex conditions provided by the matrices $\cA_\xi:=\{(A_v\, B_v\Xi_v)\}$.
\end{definition}
Notice that this rescaling allows for ``complex lengths'' of the edges. This is useful, when considering analytic properties of the Hamiltonian with respect to the parameters.

The same arguments as in proving Theorem \ref{T:domain_A} lead to the following statement:
 \begin{theorem}\label{T:domain_AXi}
\indent
\begin{enumerate}
\item The bundle $\Domain$ over $U\times(\CC\setminus\{0\})^E$ with the fiber over $(\cA,\xi)$ defined by the vertex condition matrix $\cA_\xi$, is an analytic sub-bundle of co-dimension $2E$ of the ambient trivial bundle \begin{equation*}
    \left(U\times(\CC\setminus\{0\})^E\right)\times \widetilde{H}^2(\G)\mapsto \left(U\times(\CC\setminus\{0\})^E\right).
    \end{equation*}
\item The bundle $\Domain$ is trivializable. In other words, there exists a trivialization, i.e. an analytic operator-function $T(\cA,\xi)$ on $U\times(\CC\setminus\{0\})^E$ with values in linear bounded operators from a Hilbert space $H$ into $\widetilde{H}^2(\G)$, such that $\ker T(\cA,\xi)=0$ and $\mathrm{ran}\, T(\cA,\xi)=\Domain_{\cA_\xi}$ for any $(\cA,\xi)\in U\times(\CC\setminus\{0\})^E$.
\item After the trivialization, the values of the resulting analytic in $U\times(\CC\setminus\{0\})^E$ operator-function
    \begin{equation*}
    \cA\mapsto \Hamiltonian_{\cA,\xi}T(\cA,\xi)
    \end{equation*}
    are Fredholm operators of index zero from $H$ to $L^2(\G)$.
\end{enumerate}
\end{theorem}

\subsection{Dependence of the spectrum on the vertex conditions}
\label{S:spect_perturb}

In this section, we will derive some basic facts about relations
between the quantum graph eigenvalues and continuous graph parameters
(such as matrices $\cA=\{(A_v\, b_v)\}|_{v\in\Vertices}$ of vertex
conditions and edges' lengths $\{L_e\}|_{e\in\Edges}$. We
denote by $\Hamiltonian_\cA$ the operator $-d^2/dx^2$ on a compact
graph $\G$ with the domain described by the vertex conditions
(\ref{E:cond_kossch}) that correspond to the matrix $\cA=\{(A_v\,
b_v)\}|_{v\in\Vertices}\subset U$. Here $U\subset \CC^{2q}$ (with
$q=\sum d_v^2$) consists of such complex $\cA$ that the rank of the
$d_v\times 2d_v$-matrix $(A_v\, B_v)$ is maximal for any vertex
$v$. We will also use notation $\Hamiltonian_{\cA,\xi}$ for the
re-scaled version of $\Hamiltonian_{\cA}$ that acts as
$-\xi^{-2}d^2/dx^2_e$ on the domain described by $\cA_\xi=(A_v\,
B_v\Xi_v)$. Here $\xi=\{\xi_e\}\in(\CC\setminus\{0\})^E$ is
the vector of (non-zero) scaling factors applied on each edge.

We are interested in the dependence of the spectrum of
$\Hamiltonian_{\cA,\xi}$ on the parameters $(\cA,\xi)\subset U\times
(\CC\setminus \{0\})^E$. Thus, we consider the operator pencil
$\Hamiltonian_{\cA,\xi}-\lambda\Id$ of unbounded operators in
$L^2(\G)$.

The following result, which addresses analytic behavior of the
spectrum, follows from Theorem \ref{T:domain_AXi} and the analytic
Fredholm theorem (Appendix \ref{A:Fredholm}) and \cite[Theorem 4.11]{ZaiKreKuc_umn75}.

\begin{theorem}\label{T:analit}\indent
  \begin{enumerate}
  \item The set $\mathcal{S}$ of all vectors 
    \begin{equation*}
      (\cA,\xi,\lambda)\subset U\times(\CC\setminus\{0\})^E\times\CC
    \end{equation*}
    such that $\Hamiltonian_{\cA,\xi}-\lambda\Id$ does not have a
    bounded inverse in $L^2(\G)$, is principal analytic. Namely, there
    exists a non-zero function $\Phi(\cA,\xi,\lambda)$ analytic in
    $U\times(\CC\setminus\{0\})^E\times\CC$, such that $\mathcal{S}$
    coincides with the set of its zeros.

  \item For any integer $k\geq 1$, the set  $\mathcal{S}_k$ of all vectors
    \begin{equation*}
      (\cA,\xi,\lambda)\subset U\times(\CC\setminus\{0\})^E\times\CC
    \end{equation*}
    such that $\dim \mbox{Ker}
    \left(\Hamiltonian_{\cA,\xi}-\lambda\Id\right) \geq k$, is
    analytic.
  \end{enumerate}
\end{theorem}

\begin{remark}
  \indent
  \begin{itemize}
  \item The set $\mathcal{S}$ is the graph of the multiple-valued function
    \begin{equation*}
      (\cA,\xi) \to \sigma(\Hamiltonian_{\cA,\xi}),
    \end{equation*}
    and can be considered as a kind of ``dispersion relation.'' Here
    $\sigma(\Hamiltonian)$ denotes the spectrum of the operator
    $\Hamiltonian$.
  \item The sets $\mathcal{S}_k$ are clearly nested:
    $\mathcal{S}_{k+1}\subset \mathcal{S}_k$. Moreover,
    $\mathcal{S}_1=\mathcal{S}$.
  \item For $\cA\in U_s$ (i.e., satisfying the condition guaranteeing
    self-adjointness of $\Hamiltonian_\cA$) and $\xi\in (\R\setminus
    \{0\})^E$, the operator $\Hamiltonian_{\cA,\xi}$ is selfadjoint
    and its spectrum is real and discrete.
  \end{itemize}
\end{remark}

\subsection{Eigenfunction dependence}

We present here, for completeness, a simple and well known consequence
of perturbation theory.
\begin{theorem}
  The kernel
  $\mbox{Ker}\left(\Hamiltonian_{\cA,\xi}-\lambda\Id\right)$ forms a
  holomorphic $k$-dimensional vector bundle over the set
  $\mathcal{S}_k\setminus \mathcal{S}_{k+1}$. In other words, if one
  has an local analytic eigenvalue branch $\lambda(\cA,\xi)$ of
  constant multiplicity $k$, then, at least locally, one can choose an
  analytic basis of $k$ eigenfunctions.
\end{theorem}
\begin{proof} 
  The spectral projector onto the corresponding spectral subspace is
  clearly analytic with respect to the parameter.
\end{proof}

\section{An Hadamard type formula}
\label{sec:Hadamard}

Hadamard's variational formulas
\cite{Had_mas08,GarSch_jam53,IvaKotKre_dup77,Pee_jde80} deal with the
variation of the spectral data with respect to the domain
perturbation.  For simplicity, we will consider the case of variation
of the eigenvalues with respect to a change in a ``loose'' edge's
length.  Namely, the end of the edge is assumed to be a vertex of
degree $1$ and the length of the edge will be denoted by $s$.  We will
also use $s$ to denote the vertex of degree 1.  The proof follows the
well established pattern (see, e.g. \cite{Gri_jota10}) and can be
easily generalized.

\begin{proposition}
  \label{prop:Hadamard_formula}
  Let $\lambda = \lambda(s)$ be a simple eigenvalue of a graph
  $\Gamma$ with a loose edge of length $s$ and the Dirichlet condition
  imposed at the end-vertex $s$.  Let $f=f_s(x)$ be the corresponding
  normalized eigenfunction.  Then
  \begin{equation}
    \label{eq:Hadamard_formula}
    \frac{d \lambda}{ds} = - \left|f'(s)\right|^2,
  \end{equation}
  where $f'(s)$ is the value of the derivative of the eigenfunction at
  the end-vertex.
\end{proposition}

\begin{proof}
  First we remark that by preceding theorems we can differentiate both the eigenvalue and
  the eigenfunction.  We will omit the subscript $s$ of $f$ unless we
  want to highlight the dependence of the eigenfunction on $s$.

  The Dirichlet condition at the vertex $s$ has the form
  \begin{equation*}
    f_s(s) = 0.
  \end{equation*}
  Differentiating it with respect to $s$ we get
  \begin{equation}
    \label{eq:two_derivatives}
    \frac{\partial f}{\partial s} + f'(s) = 0.
  \end{equation}
  On the other hand, the eigenfunction is $L_2$-normalized.
  Differentiating the normalization condition $(f,f)=0$ we get
  \begin{equation*}
    2 \Real \left( f, \frac{\partial f}{\partial s} \right) = 0,
  \end{equation*}
  where we used the fact that the loose edge's contribution to the
  derivative is
  \begin{equation}
    \label{eq:f_df_orthogonal}
    \frac{\partial}{\partial s} \int_0^s |f_s(x)|^2 dx 
    = \left|f(s)\right|^2 + 2\Real \int_0^s f \overline{\frac{\partial f}{\partial
        s}} dx
  \end{equation}
  and $f_s(s)=0$.
  Finally, representing the eigenvalue as $\lambda = h[f,
  f]$, where $h$ is the quadratic form~\eqref{E:qform}, we get
  \begin{equation}
    \label{eq:d_lambda}
    \frac{d \lambda}{ds} = \left|f'(s)\right|^2
    + V(s) \left|f(s)\right|^2
    + 2\Real h\left[f, \frac{\partial f}{\partial s}\right].
  \end{equation}
  To evaluate the sesqui-linear form $h$ we note that the derivative
  $\partial f / \partial s$ satisfies the same vertex conditions as
  $f$ everywhere apart from the point $s$.  Integrating by parts we get
  \begin{equation*}
    h\left[f, \frac{\partial f}{\partial s}\right] 
    = \left(\Hamiltonian f, \frac{\partial f}{\partial s}\right) +
    f'(s)  \overline{\frac{\partial f}{\partial s}}(s).
  \end{equation*}
  Now we use that $\Hamiltonian f = \lambda f$ and
  equations~\eqref{eq:two_derivatives} and \eqref{eq:f_df_orthogonal}
  to get
  \begin{equation*}
    2 \Real h\left[f, \frac{\partial f}{\partial s}\right] 
    = 2\lambda \Real \left(f, \frac{\partial f}{\partial s}\right) - 
    2|f'(s)|^2 = -2|f'(s)|^2.
  \end{equation*}
  Substituting the last equation and the Dirichlet condition into
  \eqref{eq:d_lambda} yields the desired result.
\end{proof}

One can obtain similar results when varying vertex conditions rather
than the length.  For simplicity we only consider the $\delta$-type
vertex conditions.

\begin{proposition}
  \label{prop:delta_deriv}
  Let $\lambda=\lambda(\alpha)$ be a simple eigenvalue of a graph
  $\Gamma$ which satisfies $\delta$-type vertex condition at $v$ with
  the parameter $\alpha\neq\infty$.  Then
  \begin{equation}
    \label{eq:alpha_deriv}
    \frac{d\lambda}{d\alpha} =  |f(v)|^2.
  \end{equation}
  If we re-parameterize the conditions at $v$ as
  \begin{equation}
    \label{eq:delta_zeta}
    \zeta \sum\limits_{e \in \Edges_{v}} \frac{df}{dx_e}(v)= -f(v),
  \end{equation}
  now allowing Dirichlet ($\zeta=0$) and excluding Neumann
  ($\zeta=\infty$) conditions, the derivative is
  \begin{equation}
    \label{eq:zeta_deriv}
    \frac{d\lambda}{d\zeta} =  \left| \sum\limits_{e \in \Edges_{v}}
      \frac{df}{dx_e}(v) \right|^2.
  \end{equation}
\end{proposition}

\begin{proof}
  The proof follows the pattern of the proof of
  Proposition~\ref{prop:Hadamard_formula} with minor modifications.
  We deal with $\alpha$-derivative first.  The derivative of the
  normalization condition is now
  \begin{equation*}
    2\Real \left( f, \frac{\partial f}{\partial \alpha} \right) = 0.
  \end{equation*}
  Taking the derivative of the quadratic form, see
  \eqref{E:qform_delta}, yields
  \begin{equation*}
    \frac{d\lambda}{d\alpha} = |f(v)|^2 + 2\Real h\left[f, \frac{\partial f}{\partial \alpha}\right].
  \end{equation*}
  Integrating by parts inside the sesqui-linear form $h$ and
  collecting together all the ``boundary'' terms at $v$, we get 
  \begin{equation*}
    h\left[f, \frac{\partial f}{\partial \alpha}\right] 
    = \alpha \overline{\frac{\partial f}{\partial \alpha}(v)} f(v)
    - \overline{\frac{\partial f}{\partial \alpha}(v)} 
    \sum\limits_{e \in \Edges_{v}} \frac{df}{dx_e}(v)
    + \left(\Hamiltonian f, \frac{\partial f}{\partial \alpha}\right)
    = \left(\Hamiltonian f, \frac{\partial f}{\partial \alpha}\right),
  \end{equation*}
  where we used the $\delta$-type condition at the vertex $v$.  Now we use
  $\Hamiltonian f = \lambda f$ and the derivative of the normalization
  condition to get
  \begin{equation*}
    2\Real h\left[f, \frac{\partial f}{\partial \alpha}\right]  = 0,
  \end{equation*}
  and, therefore,
  \begin{equation*}
    \frac{d\lambda}{d\alpha} = |f(v)|^2.
  \end{equation*}
  To deal with the Dirichlet case, we calculate, for $\alpha\neq0$ and
  $\zeta\neq 0$,
  \begin{equation*}
    \frac{d\lambda}{d\zeta} = \frac{1}{\zeta^2}
    \frac{d\lambda}{d\alpha} = \frac{|f(v)|^2}{\zeta^2} 
    = \left| \sum\limits_{e \in \Edges_{v}} \frac{df}{dx_e}(v) \right|^2,
  \end{equation*}
  using condition~\eqref{eq:delta_zeta} in the last step.  Since
  $\lambda(\zeta)$ is an analytic function, the value of the
  derivative at $\zeta=0$ now follows by continuity.
\end{proof}

\section{Eigenvalue interlacing}
\label{sec:eig_interlacing}

We will be assuming that the Hamiltonian $\Hamiltonian$ is
$-d^2/dx^2+V(x)$ with the vertex conditions specified in the
results. The eigenvalues $\lambda_n$ of $\Hamiltonian$ (also referred
to as the eigenvalues of the graph and denoted $\lambda_n(\Graph)$)
are labeled in non-decreasing order, counting their multiplicity.

The first theorem of this section describes the effect of modifying
the vertex condition at a single vertex.  We denote by
$\Graph_{\alpha}$ a compact (not necessarily connected) quantum graph
with a distinguished vertex $v$. Arbitrary self-adjoint conditions are
fixed at all vertices other than $v$, while $v$ is endowed with the
$\delta$-type condition with coefficient $\alpha$:
\begin{equation*}
  \left\{
    \begin{array}{l}
      f \mbox{ is continuous at } v \mbox{ and}\\
      \sum\limits_{e \in \Edges_{v}} \frac{df}{dx_e}(v)=\alpha f(v).
    \end{array}
  \right.
\end{equation*}
  
\begin{theorem}
  \label{thm:interlacing}
  Let $\Graph_{\alpha'}$ be the graph obtained from the graph
  $\Graph_{\alpha}$ by changing the coefficient of the condition at
  vertex $v$ from $\alpha$ to $\alpha'$.  If $-\infty < \alpha <
  \alpha' \leq \infty$ (where $\alpha' = \infty$ corresponds to the
  Dirichlet condition, see section~\ref{sec:conditions}), then
  \begin{equation}
    \label{eq:interlacing_monotone}
    \lambda_n(\Graph_\alpha) \leq \lambda_n(\Graph_{\alpha'}) \leq
    \lambda_{n+1}(\Graph_\alpha).
  \end{equation}
  If the eigenvalue $\lambda_n(\Graph_{\alpha'})$ is simple and it's
  eigenfunction $f$ is such that either $f(v)$ or $\sum f'(v)$ is
  non-zero then the inequalities can be made strict,
  \begin{equation}
    \label{eq:interlacing_strong}
    \lambda_{n}(\Graph_{\alpha}) < \lambda_n(\Graph_{\alpha'})
    < \lambda_{n+1}(\Graph_{\alpha}).
  \end{equation}
\end{theorem}

\begin{proof}
  The case of strict inequalities follows simply from the positivity
  of the derivative of $\lambda_n(\Graph_{\alpha'})$ with respect to
  the parameter of the vertex condition,
  Proposition~\ref{prop:delta_deriv}.  For the possibly degenerate
  case we directly use the monotonicity of the quadratic form and
  rank-one nature of the perturbation.

  Denoting by $\Graph_{\infty}$ the graph with the Dirichlet condition
  at the vertex $v$, we will actually prove the chain of inequalities
  \begin{equation}
    \label{eq:interlacing_monotone_mod}
    \lambda_n(\Graph_\alpha) \leq \lambda_n(\Graph_{\alpha'}) \leq
    \lambda_n(\Graph_{\infty}) \leq \lambda_{n+1}(\Graph_\alpha),
  \end{equation}
  which is obviously equivalent to inequality
  (\ref{eq:interlacing_monotone}).  Since we are now considering the
  Dirichlet case separately, we will assume that $\alpha'\neq\infty$.

  Consider the quadratic forms $h_\alpha$, $h_{\alpha'}$ and
  $h_\infty$ of the corresponding Hamiltonians.  According to the
  discussion in section~\ref{sec:conditions}, we have
  \begin{equation*}
    h_\alpha[f,f] = h_\infty[f,f] + \alpha|f(v)|^2 \quad
    \mbox{and} \quad
    h_{\alpha'}[f,f] = h_\infty[f,f] + \alpha' |f(v)|^2
  \end{equation*}
  on the appropriate subspaces of $\widetilde{H}^1(\G)(=\bigoplus_{e\in\Edges}H^1(e))$.  In
  fact, $D(h_\alpha) = D(h_{\alpha'})$ and $D(h_\infty) = \{f \in
  D(h_\alpha): f(v) = 0\}$.

  All inequalities follow from the min-max description of the
  eigenvalues, namely
  \begin{equation}
    \label{eq:minimax}
    \lambda_n = \min_{X \subset D:\, \dim(X)=n}\ \max_{f\in X:\, \|f\|=1} h[f,f].
  \end{equation}

  The first inequality in (\ref{eq:interlacing_monotone_mod})
  follows immediately from the observation that $h_{\alpha'} \geq
  h_\alpha$ for all $f$.

  The domain $D(h_\infty)$ is smaller than $D(h_{\alpha'})$ and the
  forms $h_\infty$ and $h_{\alpha'}$ agree on $D(h_\infty)$.
  Minimization over a smaller space results in a larger result, implying
  the second inequality in (\ref{eq:interlacing_monotone_mod}).

  The last inequality follows from the fact that $D(h_\infty)$ is a
  co-dimension one subspace of $D(h_{\alpha'})$.  To provide more
  detail, let the minimum for $\lambda_{n+1}(\Graph_\alpha)$ be
  achieved on the subspace $Y$ (which is the span of the first $n+1$
  eigenvectors) of dimension $n+1$. Then there is a subspace
  $Y_\infty$ of dimension $n$, such that $Y_\infty \subset Y$ and
  $Y_\infty\subset D(h_\infty)$.  Then
  \begin{align*}
    \lambda_n(\Graph_\infty)
    &\ =\  \min_{X:\, \dim(X)=n}\
    \max_{f\in X} h_\infty
    \ \leq\  \max_{f\in Y_\infty} h_\infty \\
    &\ =\  \max_{f\in Y_\infty} h_\alpha
    \ \leq\  \max_{f\in Y} h_\alpha
    \ =\  \lambda_{n+1}(\Graph_\alpha).
  \end{align*}
  This is precisely the last needed inequality.
\end{proof}

This theorem allows us to prove a simple but useful criterion for the
simplicity of the spectrum of a tree.

\begin{corollary}
  \label{cor:simple_trees}
  Let $T$ be a tree with a $\delta$-type condition at every internal
  vertex and an extended $\delta$-type condition at every vertex of
  degree 1.  If the eigenvalue $\lambda$ of $T$ has an eigenfunction
  that is non-zero on all internal vertices of $T$, then $\lambda$ is
  simple.

  Equivalently, if an eigenvalue $\lambda$ of the tree $T$ is multiple, there is an
  internal vertex $v$ such that all functions from the eigenspace of
  $\lambda$ vanish on $v$.
\end{corollary}

\begin{proof}
  The two statements are almost contrapositives of each other, modulo
  the following observation: if for every internal vertex there is an
  eigenfunction that is non-zero on it, one can construct an
  eigenfunction which is non-zero on all internal vertices.  Indeed,
  if $m$ is the dimension of the eigenspace of $\lambda$, then the
  subspace of the eigenfunctions vanishing on any given $v$ is at most
  $m-1$.  A finite union of subspaces of dimension $m-1$ cannot cover
  the entire eigenspace.

  We will work by induction on the number of internal vertices.  If a
  tree has no internal vertices, it is an interval and there is
  nothing to prove since all eigenvalues are simple.

  Assume the contrary: there is an eigenfunction $f$ which is not zero
  on all internal vertices of the tree, but $\lambda$ is not simple.
  Take an arbitrary internal vertex $v$ and another eigenfunction $g$.
  Cutting the tree at the vertex $v$ we obtain $d_v$ subtrees.  On at
  least one of them the function $g$ is not identically zero and not a
  multiple of $f$.  Let $T'$ be a such a subtree.

  Then there is an $\alpha<\infty$ such that $(\lambda,f)$ is an
  eigenpair of the tree $T'_\alpha$ endowed with the condition $f'(v)
  = \alpha f(v)$.  Similarly, there is an $\alpha'\leq\infty$ such
  that $(\lambda,g)$ is an eigenpair of the tree $T'_{\alpha'}$.  By
  inductive hypothesis, $\lambda$ is simple on $T'_\alpha$, therefore
  $\alpha' \neq \alpha$.  This, however, contradicts
  inequality~\eqref{eq:interlacing_strong}.
\end{proof}

The next theorem deals with the modification of the structure of the
graph by gluing a pair of vertices together.

\begin{theorem}
  \label{thm:interlacing_join}
  Let $\Graph$ be a compact (not necessarily connected) graph.  Let
  $v_0$ and $v_1$ be vertices of the graph $\Graph$ endowed with the
  $\delta$-type conditions, i.e.
  \begin{equation*}
    \begin{cases}
      f \mbox{ is continuous at } v_j \mbox{ and}\\
      \sum\limits_{e \in \Edges_{v_j}} \frac{df}{dx_e}(v_j)=\alpha_j f(v_j),
      \qquad j=0,1.
    \end{cases}
  \end{equation*}
  Arbitrary self-adjoint conditions are allowed at all other vertices
  of $\Graph$.

  Let $\Graph'$ be the graph obtained from $\Graph$ by gluing the
  vertices $v_0$ and $v_1$ together into a single vertex $v$, so that
  $\Edges_v=\Edges_{v_0}\cup \Edges_{v_1}$, and endowed with the $\delta$-type
  condition
  \begin{equation}\label{E:add}
    \sum\limits_{e \in \Edges_{v}} \frac{df}{dx_e}(v)
    = (\alpha_0+\alpha_1) f(v).
  \end{equation}

  Then the eigenvalues of the two graphs satisfy the inequalities
  \begin{equation}
    \label{eq:interlacing_delta}
    \lambda_n(\Graph) \leq \lambda_n(\Graph') \leq \lambda_{n+1}(\Graph).
  \end{equation}
\end{theorem}

\begin{proof}
  Similarly to the proof of Theorem~\ref{thm:interlacing}, we
  consider the quadratic forms of the two graphs and observe that they
  are defined by exactly the same expression (see
  section~\ref{sec:conditions}).  However, joining the vertices
  together imposes an additional restriction on the domain of the quadratic form of
  the graph $\Graph'$, namely
  \begin{displaymath}
    D(h') = \left\{ f \in D(h) :\, f(v_0)=f(v_1) \right\}.
  \end{displaymath}
  Thus, the domain $D(h')$ is a co-dimension one subspace of $D(h)$
  and the rest of the proof is identical to the proofs of the second
  and third inequalities in (\ref{eq:interlacing_monotone_mod}).
\end{proof}

Notice that if the domain of $h'$ had co-dimension $k$, then one would
have obtained by applying the same argument the inequality
\begin{equation*}
  \lambda_n(\Graph) \leq \lambda_n(\Graph') \leq \lambda_{n+k}(\Graph).
\end{equation*}
This observation immediately leads to the following generalization of
Theorem \ref{thm:interlacing_join}:

\begin{theorem}
  \label{thm:interlacing_join_k}
  Let the graph $\Graph'$ be obtained from $\Graph$ by $k$
  identifications, for example by gluing vertices $v_0$, $v_1$, \ldots
  $v_k$ into one, or pairwise gluing of $k$ pairs of vertices. Each
  identification results also in adding parameters $\alpha_j$ in the
  vertex $\delta$-type conditions, as in (\ref{E:add}). Then
  \begin{equation}
    \label{eq:interlacing_join_k}
    \lambda_n(\Graph) \leq \lambda_n(\Graph') \leq \lambda_{n+k}(\Graph).
  \end{equation}
\end{theorem}
This statement can be also proved by the repeated application of
Theorem~\ref{thm:interlacing_join}.

\begin{remark}
  Theorem~\ref{thm:interlacing} applied to the case
  $\alpha'=\infty$ is also a result about joining $d_v$ vertices
  together, since the vertex Dirichlet condition has the effect of
  disconnecting the edges at the vertex.  Note the difference with the
  result in Theorem~\ref{thm:interlacing_join_k}: the eigenvalues of
  the joined graph are now shifted down, but not further than the next
  eigenvalue of $\Graph$, which contrasts with the weaker ``not
  further than $k$-th next eigenvalue'' result of the Theorem
  \ref{thm:interlacing_join_k}.
\end{remark}

\section{Some applications}
\label{sec:applications}

\subsection{Dependence of the spectrum on the coupling constant $\alpha$ at one vertex}
\label{sec:alpha_dep}

As in section~\ref{sec:eig_interlacing} we consider the family of compact graphs $\Graph_\alpha$ with a distinguished vertex $v$.
Arbitrary self-adjoint conditions are fixed at all vertices other than
$v$, while $v$ is endowed with the $\delta$-type condition with
coefficient $\alpha$:
\begin{equation*}
  \left\{
    \begin{array}{l}
      f \mbox{ is continuous at } v \mbox{ and}\\
      \sum\limits_{e \in \Edges_{v}} \frac{df}{dx_e}(v)=\alpha f(v).
    \end{array}
  \right.
\end{equation*}
We will be using the second condition in the form 
%
\begin{equation}\label{E:extended_compl_again}
  (z+1)\sum_{e \in \Edges_v} \frac{df}{dx_e}(v)
  = \rmi (z-1) f(v),
\end{equation}
where $|z|=1$. The Dirichlet condition corresponds to $z=-1$, while
the Neumann-Kirchhoff corresponds to $z=1$.

We denote by $\cH(z)$ the operator with the above condition at $v$ and
the previously fixed set of conditions at all other vertices.  Then
main result of this section is representation of the spectrum of this
operator as the range at $z$ of an irreducible multiple valued
analytic function defined near the unit circle, plus a fixed discrete
set, see Fig.~\ref{fig:spectral_spiral}.

\begin{figure}[th]
  \centering
  \includegraphics{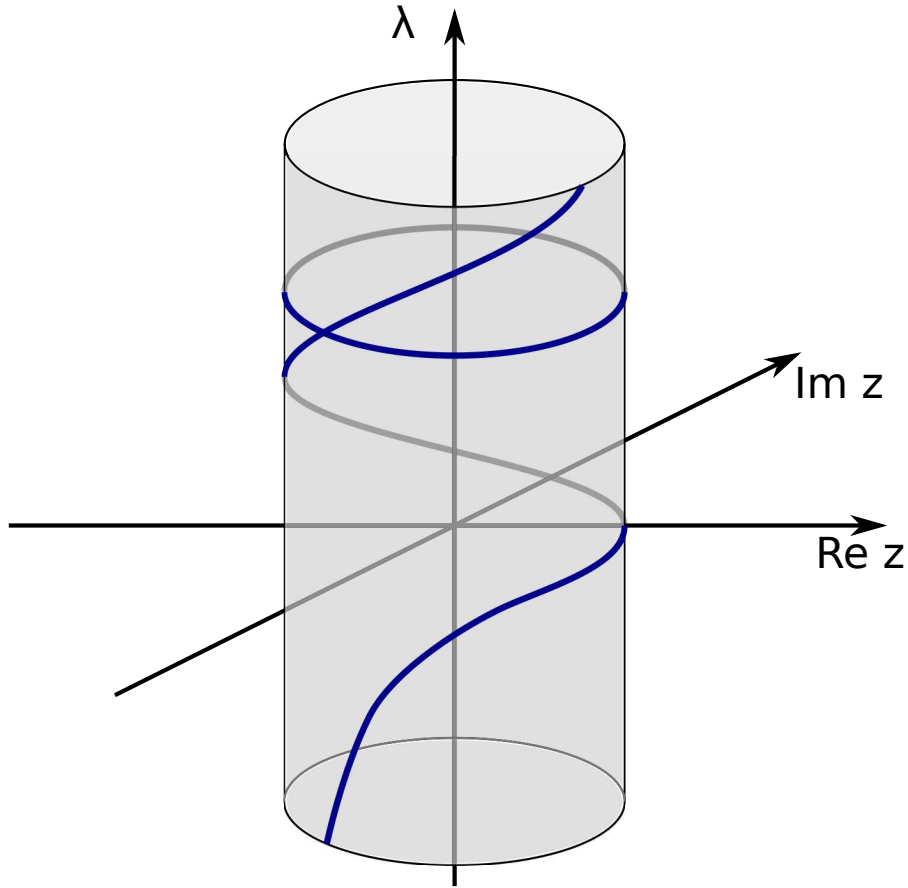}
  \caption{Spectrum of a graph as a function of the parameter $z$,
    which controls the matching condition of the form
    \eqref{E:extended_compl_again} at one of the vertices of the graph.  The
    figure shows the non-decreasing function $\Lambda(z)$ spiraling
    upwards and one of the constant branches from the set $\Delta$.}
  \label{fig:spectral_spiral}
\end{figure}

\begin{theorem}
  There exist a bounded from below discrete set $\Delta\in\R$ and a
  multiple valued function $\Lambda(z)$ (``dispersion relation'')
  analytic in a neighborhood of the unit circle $S^1$ and real on
  $S^1$, such that:
  \begin{enumerate}
  \item For any $z\in S^1$, one has
    \begin{equation}\label{E:dispersion}
      \sigma(\cH(z))=\Lambda(z)\bigcup\Delta,
    \end{equation}
    where $\sigma(A)$ denotes the spectrum of the operator $A$.
  \item The function $\Lambda$ is irreducible, i.e. any of its
    branches determines by analytic continuation the whole function
    $\Lambda$.
  \item Each analytic branch of $\Lambda$ is monotonically increasing
    in the counter-clockwise direction of $S^1$.
  \end{enumerate}
\end{theorem}

\begin{proof}
Notice that, as we already know, for each $z\in S^1$, $\cH(z)$ is a bounded from below self-adjoint operator with a discrete spectrum. Let $\lambda_n(z)$ be its $n$th  eigenvalue, counted with multiplicity in non-decreasing order. Then, according to the perturbation theory, it is continuous on the unit circle cut at $z=-1$, where one can observe that the ground state $\lambda_1$ has an one-side limit equal to $-\infty$.
\begin{lemma}\label{L:monot}\indent
\begin{enumerate}
\item The function $\lambda_n(z)$ is non-decreasing in the counter-clockwise direction on the unit circle.
\item There exists a function $F(\lambda,z)$ analytic in a neighborhood of the unit circle $S^1$ such that
$\lambda\in\sigma(\cH(z))$ if and only if $F(\lambda,z)=0$.
\end{enumerate}
\end{lemma}
Indeed, the first claim of the Lemma is a rephrasing of a part of
Theorem~\ref{thm:interlacing} (and, essentially, a consequence of
monotonicity of the quadratic form (\ref{E:quadr}) with respect to
$\alpha$).  The second claim is a direct consequence of Theorem~\ref{T:analit}.

\begin{lemma}\label{L:multipl}
  If either
  \begin{enumerate}
  \item $\lambda$ is a multiple eigenvalue of $\cH(z)$, for some $z\in
    S^1$, or
  \item $\lambda$ belongs to the spectra of $\cH(z_1)$ and $\cH(z_2)$
    for two different points $z_1,z_2\in S^1$,
  \end{enumerate}
  then $\lambda\in\sigma(\cH(z))$ for {\bf any} $z\in S^1$.
\end{lemma}

Indeed, if the eigenspace of $\cH(z)$ corresponding to $\lambda$ is
more than one-dimensional, then it contains an eigenfunction vanishing
at the vertex $v$. Then, assuming $z\neq -1$, we conclude from
(\ref{E:extended_compl_again}) that the sum of derivatives is zero as well.
Therefore this eigenfunction works for any point $z\in S^1$.  In the
case $z=-1$ we select an eigenfunction with vanishing sum of
derivatives at $v$ and it automatically vanishes at $v$ since it
satisfies the Dirichlet condition.

Let now $\lambda\in\sigma(\cH(z_1))\bigcap\sigma(\cH(z_2))$ with the
corresponding eigenfunctions $f_1$ and $f_2$.  Without loss of
generality, $z_1 \neq -1$.  First consider the case $f_1(v)\neq 0$.
Then, by Theorem~\ref{thm:interlacing}, $\lambda$ as an eigenvalue of
$\cH(z_1)$ lies strictly between consecutive eigenvalues of $\cH(z_2)$
which contradicts our assumption.  Thus $f_1(v)=0$ and the sum of
derivatives of $f_1$ is zero by (\ref{E:extended_compl_again}).  We
conclude that $f_1$ is an eigenfunction for any $z\in S^1$ thus
completing the proof of the Lemma.

We now define the set $\Delta$ as the set of all $\lambda\in\R$ such that $\lambda\in\sigma(\cH(z))$ for all $z\in S^1$. As the previous lemma shows, one can describe $\Delta$ as the intersection $\sigma(\cH(-1))\bigcap\sigma(\cH(1))$ of the Dirichlet and Neumann spectra and thus is discrete and bounded from below.

Consider now a neighborhood $U$ in $\C$ of the circle $S^1$ and the
set $\bbar{\Sigma}$ that is the closure of the set
\begin{equation}
  \label{eq:Sigma_def}
  \Sigma := \left(\bigcup\limits_{z\in
      U}\Big(\{z\}\times\sigma(\cH(z))\Big)\right) 
  \setminus \left(\C\times \Delta\right).
\end{equation}
To put things simply, we remove from the union of the spectra all
``horizontal'' branches $(z,\lambda)$ with $\lambda\in \Delta$.  The
second statement of Lemma \ref{L:monot} then implies that the set
$\bbar{\Sigma}$ is analytic. Moreover, for the neighborhood $U$ being
sufficiently small, it is a smooth complex analytic curve. Indeed,
Lemma \ref{L:multipl} implies that for all $\lambda\notin \Delta$ the
eigenvalues are simple, and thus the eigenvalue branch is analytic.
Let now $\lambda\in\Delta$ and $(z_0,\lambda)\in \bbar{\Sigma}$ for
some $z_0\in S^1$.  Then Lemma \ref{L:multipl} says that the
eigenvalue $\lambda$ near $z_0$ either stays constant, or splits into
the constant and possibly one or more increasing branches. The
constant branches are excluded in the definition of $\bbar{\Sigma}$,
equation (\ref{eq:Sigma_def}).  Also, there can be no more than one
increasing branch, otherwise one would find two values of $z$ with the
same values of $\lambda$ in $\bbar{\Sigma}$, which the lemma allows
only to happen on horizontal branches. Moreover, as we have already
concluded, the eigenvalue on such a branch must be simple. Then
Rellich's theorem (e.g., \cite[v.4, Theorem XII.3]{ReedSimon_v14})
says that the increasing branch is analytic. We thus conclude that
$\bbar{\Sigma}$ consists of one or more non-intersecting analytic
curves.

We will show now that there is only one component, if the neighborhood
$U$ is sufficiently small. First of all, the projection onto the
$\lambda$-axis of $S_j\bigcap (S^1\times\R)$, where $S_j$ is any of
the components of $\bbar{\Sigma}$ is the whole real axis, otherwise
the component whose projection does not cover the whole real axis
would create an accumulation that would contradict the discreteness of
the spectrum for each $z\in S^1$. Then existence of two or more
components would create equal eigenvalues for at least two distinct
values of $z$, which would contradict to the Lemma \ref{L:multipl}.

Thus, $\bbar{\Sigma}$ is an irreducible analytic curve that intersects each line $\{z\}\times \R$ on the cylinder along the discrete spectrum $\sigma(\cH(z))$ and such that it is monotonically increasing counterclockwise. Hence, it forms a kind of a spiral winding around the cylinder infinitely many times when $\lambda\to\infty$ and having a vertical asymptote when $\lambda\to-\infty$ (otherwise one would get a contradiction with the boundedness of each of the spectra from below). This proves the statement of the theorem.
\end{proof}

\subsection{Nodal count on trees}
\label{sec:nodal_count}

Here we present a simple proof of a result that was first discovered in
\cite{AlO_viniti92,PokPryObe_mz96,Sch_wrcm06}.  With all
ground-setting work done in Section~\ref{sec:eig_interlacing} the
proof is significantly shortened.

\begin{theorem}
  \label{thm:nodal_count_trees}
  Let $\lambda_n$ be a simple eigenvalue of $-\frac{d^2}{dx^2} +
  V(x)$ on a tree graph $\Graph$ and its eigenfunction $f^{(n)}$ be
  non-zero at all vertices of $\Graph$. Then $f^{(n)}$ has $n-1$ zeros
  on $\Graph$.
\end{theorem}

\begin{proof}
  To simplify notation we will drop the script $n$ when talking about
  the eigen-pair $(\lambda_n, f^{(n)})$.  We will prove the result by
  induction on the number of internal vertices of the tree $\Graph$.
  If there are no internal vertices, $\Graph$ is simply an interval
  and the statement reduces to the classical Sturm's Oscillation
  Theorem (see, e.g. \cite{CourantHilbert_volume1}).

  Let $v$ be a vertex of degree $d_v$.  Then $v$ separates $\Graph$
  into $d_v$ sub-trees $\Graph_i$, each with a strictly smaller number
  of internal vertices.  For each subtree, the vertex $v$ is a vertex
  of degree 1 with, so far, no vertex condition.  We will impose
  $\delta$-type condition with the parameter $\alpha_i < \infty$
  chosen in such a way that the function $f$ restricted to $\Graph_i$
  is still an eigenfunction.  This value is simply $\alpha_i = f_i'(v)
  / f_i(v)$.  The eigenfunction $f$ still corresponds to the
  eigenvalue $\lambda_n$ which is still simple (otherwise we can
  construct another eigenfunction for the entire tree $\Graph$,
  contradicting the simplicity of $\lambda$).  It is an $n_i$-th
  eigenvalue of $\Graph_i$ and, applying the inductive hypothesis, we
  conclude that the function $f$ has $n_i-1$ zeros on the subtree
  $\Graph_i$.  Thus we need to understand the relationship between the
  numbers $n_i$ and the number $n$.

  To do it, consider the subtree $\Graph_{i,\infty}$ which now has
  Dirichlet condition on the vertex $v$.  By
  Theorem~\ref{thm:interlacing}, there are exactly $n_i-1$ eigenvalues
  of $\Graph_{i,\infty}$ that are smaller than $\lambda$.  Now
  consider the tree $\Graph_{\infty}$, which is the original tree but
  with the Dirichlet condition at the vertex $v$.  On one hand it has
  exactly $n-1$ eigenvalues that are smaller than $\lambda$.  On the
  other hand $\Graph_{\infty}$ is just a disjoint collection of
  subtrees $\Graph_{i,\infty}$ and its spectrum is the superposition
  of the spectra of $\Graph_{i,\infty}$.  Therefore the number of
  eigenvalues that are smaller than $\lambda$ is
  \begin{displaymath}
    n-1 = \sum_{i=1}^{d_v} (n_i-1).
  \end{displaymath}
  Since we already proved that the number of zeros of $f$ is equal to
  the sum on the right-hand side, we can conclude that $\mu(f)=n-1$.
\end{proof}

\section*{Acknowledgments}
The work of the first author was supported in part by the NSF grant
DMS-0907968.  The work of the second author was supported in part by
the Award No. KUS-C1-016-04, made to IAMCS by King Abdullah University
of Science and Technology (KAUST).  GB is grateful to C.R.~Rao
Advanced Institute of Mathematics, Statistics and Computer Science
(AIMSCS), Hyderabad, India, where part of the work was conducted, for
warm hospitality.  Stimulating discussions with R.~Band, U.~Smilansky
and G.~Tanner are gratefully acknowledged.

\appendix
\section{Analytic Fredholm alternative}
\label{A:Fredholm}

Here we formulate the following version of the analytic Fredholm alternative:
\begin{theorem}
Let $A(z)$ be an analytic family of Fredholm operators of index zero acting in a Banach space $E$, where $z$ runs over a holomorphically convex domain $\Omega$ in $\C^n$ (or a more general Stein analytic manifold). Then there exists an analytic function $f(z)$ in $\Omega$ such that the set of all points $z$ for which $A(z)$ is not invertible coincides with the set of all zeros of the function $f$.
\end{theorem}
The proof of this statement can be found in many places, e.g. it follows from the Corollary to Theorem 4.11 in \cite{ZaiKreKuc_umn75}.

\bibliographystyle{abbrv}
\bibliography{bk_bibl,dependence}

\end{document}